\documentclass[preprint]{elsarticle}

\usepackage{array,xspace,multirow,hhline,tikz,colortbl,tabularx,booktabs,fixltx2e,amsmath,amssymb,amsfonts,amsthm}
\usepackage{algorithm}
\usepackage{algorithmic}
\usepackage{verbatim,ifthen}
\usepackage{enumitem}
\usepackage{pifont}
\usepackage{calrsfs,mathrsfs}
\usepackage{lscape}
\usetikzlibrary{arrows}
\usetikzlibrary{positioning}
\usepackage{paralist}
\usepackage{enumitem}

	\newtheorem{lemma}{Lemma}%
	\newtheorem{theorem}{Theorem}%
	\newtheorem{example}{Example}

			\newtheorem{remark}{Remark}

	\newcommand{\clusterset}{\ensuremath{\mathcal{C}}}

\definecolor{light-gray}{gray}{0.95}

	\newcommand{\ceil}[1]{\left\lceil #1 \right\rceil}
	\newcommand{\floor}[1]{\left\lfloor #1 \right\rfloor}

	\usepackage{enumitem}
	\setenumerate[1]{label=\rm(\it{\roman{*}}\rm),ref=({\it\roman{*}}),leftmargin=*}

	\newlength{\wordlength}
	\newcommand{\wordbox}[3][c]{\settowidth{\wordlength}{#3}\makebox[\wordlength][#1]{#2}}

\usepackage{xspace}
\newcommand{\allocshares}{\textsc{AllocationFromShares}\xspace}
\newcommand{\edp}{\textsc{ExactDollarPartition}\xspace}

\usepackage{lineno}
\begin{document}

	\title{Strategyproof Peer Selection using\\ Randomization, Partitioning, and Apportionment\footnote{This is a significantly revised and expanded version of our conference paper from AAAI 2016 \cite{ALM+16a}. This version introduces the exact version of Dollar Partition along with new proofs and a new experiment.}}

		\author{Haris Aziz}\ead{haris.aziz@data61.csiro.au}
		\address{UNSW Sydney and Data61 CSIRO, Sydney 2052, Australia} 
	\author{Omer Lev} \ead{omerlev@bgu.ac.il}
	\address{Ben-Gurion University of the Negev, Beersheba 8410501, Israel}
		\author{Nicholas Mattei} \ead{nsmattei@tulane.edu}
		\address{Tulane University, New Orleans, LA 70115, USA} 
		\author{Jeffrey S. Rosenschein} \ead{jeff@cs.huji.ac.il}
	\address{Hebrew University of Jerusalem, Jerusalem 91904, Israel}	
	\author{Toby Walsh} \ead{toby.walsh@data61.csiro.au}
	\address{UNSW Sydney and Data61 CSIRO, Sydney 2052, Australia}


	\begin{abstract}

		Peer reviews, evaluations, and selections are a fundamental aspect of modern science. Funding bodies the world over employ experts to review and select the best proposals from those submitted for funding. The problem of peer selection, however, is much more general: a professional society may want to give a subset of its members awards based on the opinions of all members; an instructor for a Massive Open Online Course (MOOC) or an online course may want to crowdsource grading; or a marketing company may select ideas from group brainstorming sessions based on peer evaluation.

		We make three fundamental contributions to the study of peer selection, a specific type of group decision-making problem, studied in computer science, economics, and political science.
		First, we propose a novel mechanism that is strategyproof, i.e., agents cannot benefit by reporting insincere valuations. 
		Second, we demonstrate the effectiveness of our mechanism by a comprehensive simulation-based comparison with a suite of mechanisms found in the literature.
		Finally, our mechanism employs a randomized rounding technique that is of independent interest, as it solves the apportionment problem that arises in various settings where discrete resources such as parliamentary representation slots need to be divided proportionally.
	\end{abstract}


\maketitle

\section{Introduction}
Since the beginning of civilization, societies have been 
selecting small groups from within. Athenian 
society, for example, selected a random subset of 
citizens 
to participate in the Boule, the council of citizens that ran
daily affairs in Athens. 
Peer review, evaluation, and selection has been the main process by which scientific conferences have selected a subset of papers for publication. Increasingly, peer evaluation is
becoming popular and necessary to scale grading in MOOCs (Massive Open Online Courses, e.g., Coursera and EdX) \cite{PHCDNK13,LRP14,CKV15a}.
In all of these peer selection settings, however, we do not wish to select
an arbitrary subset of size $k$, but the ``best $k$'', and we need, therefore,
a procedure in which the candidates are rated according to the opinions of the group.
%
%
In peer selection problems we are not seeking an external, ``independent'' agent to make choices, 
but desire a crowdsourced approach, in which participants are those making the selection.
Mechanisms for peer selection and the properties of these mechanisms receive 
considerable attention within economics, political science, and computer science \cite{AFPT11a,LGE10,HG90,FeKl14a,HoMo13a,Haze13a,KLMP15b,RRS11a}.

Our initial motivation comes from the recent U.S.~National Science Foundation (NSF) ``mechanism design pilot,'' which
was an attempt to spread the review load amongst all submitters of proposals \cite{Haze13a,MeSa09a}. 
The program uses \emph{``reviewers assigned from among the set of PIs whose proposals are
being reviewed.''}
Reviewers' own proposals get \emph{``supplemented with `bonus points' depending upon the degree to
which his or her ranking agrees with the consensus ranking
(\cite{NSF14a}, Page 46).''} 
This mechanism employed by the NSF is not strategyproof; reviewers are incentivized to guess what others are thinking, not to provide their honest feedback.  Hence the mechanism induces a type of Keynesian ``Beauty Contest'' \cite{Key36a} where the incentives are misaligned and humans have been shown to not behave truthfully \cite{CGMP15a}.
 Removing the bonus may be worse, as reviewers would then be able to increase the chance of their own proposal being accepted by rating other proposals lower \cite{NaLi13a}. In either case, reviewers can \emph{benefit} from reporting something other than their true values. When agents have the incentive to misrepresent their truthful reports, the effect on the results of the aggregation or selection mechanism can be problematic. Indeed, in a comprehensive evaluation of the peer review process, Wenneras and Wold \cite{WeWo97a} wrote, \emph{``\ldots the development of peer-review systems with some built-in resistance to the weakness of human nature is therefore of high priority.''}


We propose a novel strategyproof, (which we shall also call impartial) mechanism\footnote{Strategyproof in the peer selection setting differs from the voting setting. In peer selection impartial means one cannot make themselves be selected if they wish to do so.} where
agents can never gain by being insincere.
%
There are many reasons to prefer a strategyproof mechanism:
first, the mechanism does not favor ``sophisticated'' agents
who have the expertise to behave strategically.
Second, agents with partial or no knowledge of
other agents' rankings are not at a disadvantage when using a strategyproof mechanism.
Third, normative properties of a mechanism typically assume sincere behavior on the part of the agents. If agents act strategically,
we may lose some desirable normative properties. 
Fourth, it is, in general, easier to persuade
people to use a strategyproof mechanism
than one that can be (easily) manipulated.
%
Note that while strategyproofness does not handle all potential biases of agents, it eliminates an obvious ``weakness in human nature.''

To achieve strategyproofness we could
 use a lottery (as in the Athenian democracy). 
However, 
this method does not select based on merit.
A different option is to use a mechanism based on a voting rule.
However, following
Gibbard and Satterthwaite \cite{Gib73,Sat75},
any ``reasonable'' mechanism based on voting will 
not be strategyproof unless it is a dictatorship.
Another option is to employ a mechanism like
the Page Rank algorithm that uses Markov chains to compute a ranking 
of agents \cite{Wals14a}. However, such mechanisms are
also not strategyproof. 


\paragraph{Contributions}

First, we propose a novel \emph{peer selection mechanism}, \edp, that satisfies
several desirable axiomatic properties including \emph{strategyproofness} and two natural \emph{monotonicity} properties. 
%
Second, we conduct a detailed experimental comparison with other strategyproof mechanisms with regard to their ability to recover the ``ground truth''. 
Our experiments demonstrate that \edp selects more high-quality agents more often, selects more high-quality agents in the worst case, and has more consistent quality than any other strategyproof mechanism in the literature. 
%
Third, our mechanism uses a novel \emph{randomized apportionment} subroutine
to fairly round selected fractional group sizes to integers.
This subroutine is interesting in its own right as it provides a compelling solution to the fundamental problem of apportionment: allocating representatives or resources in proportion to group size or strength of demand. 
Young \cite{Youn94a} motivates the problem as follows: \emph{``This surprisingly difficult problem has concerned statesmen, political analysts and mathematicians for over two hundred years.''}

\section{Discussion and Related Work}

Peer review is the cornerstone of modern science and hence, the quality, veracity, and accuracy of peer review and peer evaluation is a topic of interest across a broad set of disciplines. Most empirical studies of peer review and peer selection focus on the effectiveness and limits of the system, typically by assembling large corpora of peer-reviewed proposals and cross-examining them with new panels or review processes \cite{CCS81a,LiLe15a}. Questions of bias, nepotism, sexism, cronyism, among other issues, have received extensive coverage, and have been substantiated to varying degrees, in the literature \citep{CCS81a,MEFF90a,WeWo97a}. However, a consistent conclusion in the meta-research on peer review is that, in order to decrease the role of chance and/or any systematic bias, the community needs to broaden the base of reviewers. Indeed, one way for the results of the review process to reflect the views of the entire scientific constituency and provide more value to the community is to increase the number of reviewers \cite{JoRa15a,Pri14}. The key scientific question lies in finding a mechanism that allows for crowdsourcing the work of reviewing, without compromising the incentives and quality of the peer review and selection process.

The criticism that prominent peer selection mechanisms such as those under consideration by American and European funding bodies~\cite{MeSa09a,Haze13a} are \emph{not} strategyproof \cite{NaLi13a} has underscored the need to devise mechanisms with better incentive properties. The literature most directly relevant to this article is a series of papers on strategyproof (impartial) selection~\cite{HoMo13a,AFPT11a} and more recently impartial ranking \cite{KKKKP18a}.  The explosive growth in computer science and machine learning conference submissions in the past years has led to more work in the computer science and machine learning fields are using data from large conferences \cite{STMG+18a} and even performing human experiments \cite{KKPK18a} to analyze the assignment \cite{LiMaNoWa18,SSS18a} and outcomes of various novel mechanisms for peer review.
We survey and provide details of these mechanisms in the next section.
Most of the work on strategyproof peer selection focuses on the setting in which agents simply approve (nominate) a subset of agents~\cite{AFPT11a,BNV14a,FeKl14a,HoMo13a}, with the latter three of these restricting attention to the setting in which exactly one agent is selected ($k=1$).\footnote{In \citet{AFPT11a} and \citet{KLMP15b} the letter $k$ is used to denote the number of partitions, in our paper and many others, $k$ designates the number of agents selected.  Therefore we use $\ell$ to denote the number of partitions in this paper and $k$ to denote the number of agents selected.}
A popular class of strategyproof peer selection mechanisms are Partition based mechanisms, as presented in \citet{AFPT11a}, where agents are divided into non-intersecting groups.
\citet{KLMP15b} present an interesting strategyproof mechanism (Credible Subset) that performs well when each agent reviews a very small number of agents relative to the total number of agents.
Other recent work focuses on tradeoffs between different axioms concerning peer selection~\citep{BeGj14a,Mack15a}.

 Both \citet{AFPT11a} and \citet{HoMo13a} examine the selection problem in which agents simply approve (nominate) a subset of agents. \citet{HoMo13a}, \citet{FeKl14a}, and \citet{BNV14a} restrict their attention to a setting in which exactly one agent is selected ($k=1$).
 \citet{FeKl14a} also present the Permutation mechanism that achieves the same bound as the Partition mechanisms when only one agent is selected ($k=1$).
\citet{AFPT11a} and \citet{HoMo13a} showed that for the peer selection problem, deterministic impartial mechanisms are extremely limited, and must sometimes select an agent with zero nominations even though other agents receive nominations, or an agent with one nomination when another agent receives $n-1$ nominations~\cite{FeKl14a}. \citet{BFK15} built on this work to show that allowing a mechanism where agents simply approve of some subset of agents to select fewer than $k$ agents allows the mechanism to guarantee some bounds on the selected items---they are within about $1-\frac{1}{e}$ from the optimal selection. 
\citet{KLMP15b} present a more general mechanism called Credible Subset that is strategyproof but may select \emph{no winners} with non-zero probability. Credible Subset performs well when each agent reviews a few other agents, and this number is considerably smaller than $k$. 

There are a number of practical application areas that are related to and/or use peer selection.
The peer selection problem is closely related to peer-based grading/marking~\cite{DeSh13a,JoRa15a,KWLC13a,PHCDNK13,Robi01a,Wals14a,WTL15a} especially when students are graded
based on percentile scores.
For peer grading, mechanisms have been proposed that make a student's grade slightly dependent on the student's grading accuracy (see e.g., \citet{Wals14a} and \citet{MeSa09a}). However such mechanisms are not strategyproof as one may alter one's reviews to obtain a better personal grade.
Finally, as an additional and recent application area, economists have studied mechanisms and the strategic issues that arise in using peer evaluation for micro-financing and other reputation based resource allocation problems \cite{Baum18a,BlOl18a,HRR18a}.

\section{Setup and Survey of Existing Mechanisms}\label{mechSurvey}


Given a set $N$ of agents $\{1,\ldots, n\}$ where each agent, depending on the setting, evaluates some $m$ of the other agents where $0 \leq m \leq n-1$.
Each agent reports a valuation (review) over the other agents (proposals). These reports could be cardinal valuations $v_i(j)$ for agent $i$'s valuations of agent $j$, or they could be a weak order reported by agent $i$ of agents in $N\setminus \{i\}$, which may be transformed to cardinal valuations using a scoring rule. Based on these reported evaluations, around $k$ agents are selected.
Some mechanisms, such as Credible Subset, 
may not always return a size of exactly $k$ even if the \emph{target} size is $k$.


A particular family of mechanisms with which we will deal are based on \emph{partitioning}. The general idea of partitioning-based mechanisms is to divide the agents into a set of clusters $\clusterset = \{C_1, \ldots, C_{\ell}\}$. This partition can be done using either a random process or some predetermined process that does not include randomization.
We will assume that cluster sizes are such that selection from them is not a problem: for all $1\leq i\leq \ell$, $k\leq |C_{i}|$. If $\frac{N}{\ell}$ is not an integer then we assume that $k \leq \floor{\frac{N}{\ell}}$, the smallest cluster size.

\subsection{Mechanisms}
There are three prominent mechanisms for peer selection that appear in the literature.

\begin{description}[itemsep=0.1cm,leftmargin=0.25cm]
\item[Vanilla:] 
Select the $k$ agents with the highest total value based on their reviews by other agents (as done today, for example, in many scientific conferences).
Vanilla is not strategyproof; unselected agents have an incentive to lower their reported valuations of selected agents. 
\item[Partition:] Divide the agents into ${\ell}$ clusters and select a preset number of agents from each cluster, typically $k/\ell$ (rounded in some way if $k/\ell$ is not an integer), according to the valuations of the agents \emph{not} in that cluster. This class of mechanisms is a straightforward generalization of the Partition mechanism \citep{AFPT11a,FeKl14a} (and in an early version of \citet{KLMP15b}) which is strategyproof. 
\item[Credible Subset \cite{KLMP15b}:]
Let $T$ be the set of agents who have the top $k$ scores, as in Vanilla. 
Let $P$ be the set of agents who do not have the top $k$ scores but will make it to the top $k$ if they do not contribute any score to other agents (hence $|P|\leq m$). 
With probability $(k+|P|)/(k+m)$, Credible Subset selects a set of $k$ agents uniformly at random from $T\cup P$, and with probability $1-(k+|P|)/(k+m)$, it selects no one. The mechanism is strategyproof. 


%
%
%
%
%
\end{description}
There are a number of other mechanisms that are tailor-made for $k=1$ and when agents only mark approval of a subset of agents: Partition~\citep{HoMo13a}; Permutation~\citep{FeKl14a}; and Slicing~\citep{BNV14a}.
%
%
%
%
%
%
%
%
When designing our mechanism, we were inspired by mechanisms for dividing a \emph{continuous resource} from the economics literature~\cite{CMT08a,TiPl08a}. In particular, we use ideas from the following mechanism.

\begin{description}[itemsep=0.1cm,leftmargin=0.25cm]
\item[Dividing a Dollar:] Each agent $i$ reports a value $v_{i}(j)$ that is his estimation of how much of the resource agent $j$ should receive. These values are normalized so that $\sum_{j\in N\setminus \{i\}}v_i(j)=1/n$. Hence, the \emph{Dollar share} of each agent $i$ is $x_i=\sum_{j\in N\setminus \{i\}}v_j(i)$.
 \end{description}



\subsection{Properties of Mechanisms}

We consider some basic axioms of peer selection mechanisms. When algorithms involve randomization, these properties are with regard to the probability of selection.
\begin{description}
	\item[Anonymity:] For some permutation of $n$ agents $\pi$, if $W$ is the outcome of the mechanism for agents $N$, with each agent $i$ giving a valuation on agents $i^{1},\ldots i^{m}$, then $\pi(W)$ is the outcome of the mechanism for agents $N$ where each agent $\pi(i)$ gives valuations on agents $\pi(i^{1}),\ldots,\pi(i^{m})$.
	\item[Non-imposition:] For any target set $W$, there is a valuation profile and a randomization seed that achieves $W$.
	\item[Strategyproofness (Impartiality):] Agents cannot affect their own selection. 
	\item[Monotonicity:] 
	If agent $i$ is selected, and some other agent $j$ reinforce it, increasing $i$'s relative position in her ranking without changing the relative position of other agents, then agent $i$ will still be selected. \footnote{When scores are used instead of ordinal rankings, we are, in a sense, converting them to ordinal rankings by looking at normalized scores, in which the sum of all scores is $1$. The property states that if only agent $i$'s normalized score is raised, it will still be selected.}

	\item[Committee Monotonicity:] If $W$ is the outcome when the target set size is $k$, then all the agents in $W$ are still selected if the target set size is $k+1$. 
	\bigskip
	
\end{description}

\section{\edp}

The algorithm \edp is formally described in Algorithm \ref{algo:DP}. 
Broadly, it works as follows: agents are partitioned into $\ell$ clusters such that the sizes of clusters are equal or as near as possible, with difference at most $1$. Each agent $i\in N$ assigns a value $v_i(j)$ to each agent $i'$ that is among the $m$ agents that $i$ reviews, none of which are in $i$'s cluster. Agent $i$ may directly give a cardinal value to the agents they review or the cardinal value may be obtained by a scoring function that converts an ordinal ranking given by $i$ to cardinal values. In either case, the values that $i$ gives are normalized so that agent $i$ assigns a total value of 1 to the $m$ agents they are to review outside their own cluster. Based on the values from agents outside the cluster we assign the normalized weight (Dollar Share) $x_j$ to each cluster $C_j$. Based on each Dollar Share $x_j$, each cluster has a  quota $s_j=x_j \cdot k$ (possibly real but not rational). 
If all $s_j$'s are integers, then each $s_j$ is the quota of cluster $C_j$, i.e., the top graded $s_j$ agents are selected from cluster $C_j$. 
If not all $s_j$ are integers, then we use the function \allocshares in line 7 (detailed in the next section) to enumerate discrete cluster allocations in which each cluster gets an allocation of either $\floor{s_{j}}$ or $\ceil{s_{j}}$. \allocshares then computes a probability distribution over these discrete allocations, requiring at most $\ell$ such allocations, so that the expected quota for each cluster will be exactly $s_{j}$. We draw a discrete allocation $(t_{1},\ldots,t_{\ell})$ using the distribution computed in \allocshares and select exactly the $t_j$ agents from each cluster $C_j$ with the highest score.
The agents who have a higher score will be referred to as having a higher ranking. We note that the algorithm gracefully handles the case where an agent is absent, i.e., does not submit their reviews, or if she gives zero score to every other agent that she is responsible for reviewing. The algorithm handles this case by forcing the agent to give equal score to all the other agents reviewed, i.e., $\frac{1}{m}$.



		\begin{algorithm}[ht!]
		 \caption{\edp}
		 \label{algo:DP}
		\renewcommand{\algorithmicrequire}{\wordbox[l]{\textbf{Input}:}{\textbf{Output}:}}

		 \renewcommand{\algorithmicensure}{\wordbox[l]{\textbf{Output}:}{\textbf{Output}:}}
		\algsetup{linenodelimiter=\,}
		\begin{algorithmic}
			\REQUIRE Set of agents $N$, valuations $(v_1,\ldots, v_n)$ of the agents, $m$ the number of reviews per agent, and $\ell$ the number of clusters.
			\ENSURE Set of winning agents $W$.
		\end{algorithmic}
		 \begin{algorithmic}[1] 
		\STATE Initialize $W\gets \emptyset$
			\STATE Generate a partition $\{C_1,\ldots, C_{\ell}\}$ of $N$ where the difference between the sizes of any two clusters is at most 1.
			\STATE Each $i\in N$ reviews $m$ agents outside $C(i)$, where $C(i)$ is the cluster of agent $i$, so that any reviewed agent $j$ is assigned a valuation $v_{i}(j)$.
			
			\STATE Ensure $\sum_{j\notin C(i)}v_i(j)=1$ by normalizing. If $v_i(j)=0$ for all $j$, then we $v_i(j)=1/m$ for each $j$ reviewed by $i$. 
			 \STATE $x_i$, the value of a cluster $C_i$, is defined as:
			 \[x_i\gets \frac{1}{n} \times \sum_{j\in C_i, j'\notin C_i}v_{j'}(j).\]
			\COMMENT {Using the $x_i$ values, we now compute the number of agents $t_i$ to be chosen from each cluster $C_i$.}
			 \STATE Let each share $s_i\gets x_i \cdot k \text{ for each } i\in \{1,\ldots, \ell\}$.
		\STATE $(t_{1},\ldots,t_{\ell})\gets$ \allocshares({$s_{1},\ldots,s_{\ell}$})
		where $(t_{1},\ldots,t_{\ell})$ are the number of agents to be allocated from each cluster.

								\STATE For each $i\in C(i)$, the score of agent $i$ is
								 $\sum_{i'\notin C(i)}v_{i'}(i).$
								 \STATE Select $t_j$ agents with the highest scores from each cluster $C_j$ and place them in set $W$.
							
							\RETURN $W$
		 \end{algorithmic}
		\end{algorithm}

 We illustrate the working of our algorithm with the following example.

\begin{example}
Suppose we want to select $k=5$ winners from our agents, which are divided into four clusters, each with 2 agents, giving us $n=8$, each agent being responsible for reviewing $m=2$ other agents. Table~\ref{example1} shows the initial grades, on a scale of $0-100$ given by the row agent to their peers listed in the columns. Table~\ref{example2} shows these grades following normalization so that each agent distributes $1.0$ point to the $m=2$ agents they review.

\begin{table}[htp]
\caption{Example grades of row agent for column agent.}
\begin{center}
\begin{tabular}{l|c|c|c|c|c|c|c|c}
&A&B&C&D&E&F&G&H\\
\hline\hline
A [cluster 1] &&&&0&&&&100 \\
B [cluster 1] &&&80&&30&&& \\
C [cluster 2] &83&&&&&&42& \\
D [cluster 2] &&77&&&&50&& \\
E [cluster 3] &&&&65&&&65& \\
F [cluster 3] &&56&&&&&&98 \\
G [cluster 4] &29&&&&&62&& \\
H [cluster 4] &&&75&&29&&& \\
\end{tabular}
\end{center}
\label{example1}
\end{table}%

\begin{table}[h]
\caption{Example grades following normalization.}
\begin{center}
\begin{tabular}{l|c|c|c|c|c|c|c|c}
&A&B&C&D&E&F&G&H\\
\hline\hline
A [cluster 1] &&&&0&&&&1.00 \\
B [cluster 1] &&&0.7272&&0.2728&&& \\
C [cluster 2] &0.664&&&&&&0.336& \\
D [cluster 2] &&0.6063&&&&0.3937&& \\
E [cluster 3] &&&&0.50&&&0.50& \\
F [cluster 3] &&0.3636&&&&&&0.6364 \\
G [cluster 4] &0.3187&&&&&0.6813&& \\
H [cluster 4] &&&0.7212&&0.2788&&& \\
\end{tabular}
\end{center}
\label{example2}
\end{table}%

This means the overall scores are:
\begin{equation*}
\begin{split}
\text{Cluster 1: }&x_{1}= \frac{\sum_{j\in C_1, j'\notin C_1}v_{j'}(j)}{n} = \frac{1.9526}{8} = 0.244075\Rightarrow \\ &\text{Therefore, } s_{1}= x_1\cdot k = 1.220375\\
\text{Cluster 2: }&x_{2}= \frac{1.9484}{8} = 0.24355\Rightarrow s_{2}= 1.21775\\
\text{Cluster 3: }&x_{3}= \frac{1.6266}{8} = 0.203325\Rightarrow s_{3}=1.016625\\
\text{Cluster 4: }&x_{4}= \frac{2.4724}{8} = 0.30905\Rightarrow s_{4}=1.54525
\end{split}
\end{equation*}

This process leaves us with a share vector of $$\vec{s} = (1.220375, 1.21775, 1.016625, 1.54525).$$ While $\sum \vec{s} = k$ observe that not all the numbers are integers, leaving us the need to apportion the remainders. The function \allocshares is explored further in Example~\ref{allocationExample}, but for now, it suffices to know that since the number of agents for each cluster is rounded up or down, from one of the clusters we need to choose 2 agents, and 1 agent from the others. The highest probability is given to the event in which cluster 4 is the only one that will select 2 agents, giving us our allocation $\vec{t} = (1,1,1,2)$. This allocation vector leads to the selection of the agents $A, C, F, G, H$, which are the top-ranked agent in clusters 1, 2, and 3, and both agents of cluster 4.
\end{example}

We defer our proofs and analysis of the properties of the apportionment method--- function \allocshares---to the next section. For the analysis of the overall mechanism, it is enough to assume it chooses an allocation of size $k$ from a probability space constructed so that the expected share of each cluster $j$ is $s_{j}$. As no agent is treated differently in the mechanism, \edp is anonymous and satisfies non-imposition. 
 
%
%
%
\begin{theorem}\label{th:edp-sp}
	\edp is strategyproof.
	\end{theorem}
	
	\begin{proof}
		Suppose agent $i$ is in cluster $C_j$ of the generated partition. Agent $i$ will be selected in $W$ if and only if its score is among the top $t_j$ scores from agents in $C_j$. Therefore agent $i$ can manipulate either by increasing $t_j$ or by increasing its score relative to other agents in $C_j$ given by agents outside $C_j$. Since agent $i$ cannot affect the latter, the only way it can manipulate is by increasing $t_j$. 
		We argue that agent $i$ cannot change its \emph{expected} $t_j$ by changing its valuation $v_i$ for agents outside the cluster. Note that $i$ contributes a probability weight of $1/n$ to agents outside $C_j$ and zero probability weight to agents in $C_j$. Hence it cannot affect the value $x_j$ of cluster $C_j$. As $s_j$ is derived from $x_j$, agent $i$ cannot affect $s_j$.

As we will show in our analysis of \allocshares, specifically Theorem \ref{allocationWorks}, the expected value of $t_{j}$ will be $s_{j}$ (and its value is either $\floor{s_j}$ or $\ceil{s_j}$, in the unique probabilities that make the expected value $s_{j}$). Since any agent in cluster $C_j$ that changes its report does not affect $s_{j}$, it does not affect the expected $t_j$, nor the respective probabilities of getting $\floor{s_j}$ and $\ceil{s_j}$.
Hence, agent $i$ cannot manipulate by either increasing the $t_j$ of his cluster or by increasing his score relative to the agents in $C_j$. Therefore, \edp is strategyproof.




\end{proof}

\begin{remark}
	\edp is not just strategyproof but even group-strategyproof if manipulating coalitions involve agents from the same cluster (so potentially colluding agents, e.g., with conflict of interest, can be put in the same cluster). 
	\end{remark}
%
%
%
%
			\begin{theorem}\label{th:edp-monotonic}
				\edp is monotonic.
				\end{theorem}
				\begin{proof}
					Let us compare the valuation profile $v$ when $i$ is not reinforced and $v'$ when $i$ is reinforced.
	The relative ranking of $i$ is at least as good when $i$ is reinforced. Since any decrease in valuation that an agent $j$ in $C(i)$ receives translates into the same increase in the valuation received by agent $i$, the total valuation that $C(i)$ receives does not decrease and hence the number of agents selected from $C(i)$ is at least as high as before.
			%
\end{proof}

		\begin{theorem}\label{th:edp-commt-mono}
			\edp is committee monotonic.
			\end{theorem}
						\begin{proof}
	The only difference between running the algorithm for different target $k$ values is when calculating the quota vector $\vec{s}$. However, if agent $i$ in cluster $C_{j}$ was selected, that means its ranking in the cluster $C_{j}$ was above $t_{j}$. When $k$ increases, $s_{j}$ will only increase (as $x_{j}$ remains the same), and hence so will $t_{j}$, ensuring that $i$ will be selected again.
\end{proof}

\section{A Randomized Apportionment Rule}

The randomized allocation technique we call \allocshares used for \edp is of independent interest since it addresses the classic apportionment problem in a randomized way. Consider the problem in which $n$ agents divided into $\ell$ disjoint groups are to be allocated a given number of slots $k<n$ in proportion to the group sizes. The problem is ubiquitous in apportionment settings such as proportional representation of seats in the U.S.~congress, European Parliament, and the German Bundestag as well as various other committee selection settings~\cite{BaYo82a,BaYo80a,Birk76a,Mayb78a,Puke14a}. This problem has been studied in political science, economics, operations research, and computer science for over 200 years \cite{Youn94a}. 

In these settings, each group $i$ has a quota $s_i$ with $\sum_{i=1}^ns_i=k$, which we call its \emph{target quota}. Since $s_i$ may not be an integer, we have to resort to \emph{apportionment}, which means that in order to allocate exactly $k$ slots, some group may be assigned an integer quota slightly more or less that its target quota. 

Numerous apportionment procedures have been introduced in the literature including the methods of Hamilton, Jefferson, Webster, Adams, and Hill~\cite{BaYo82a}; each with its own drawbacks. In fact, \citet{BaYo82a} proved that \emph{no} deterministic apportionment procedure can satisfy a group of three minimal axioms: (1) \emph{Quota Rule}, each group should get quota that is the result of the target quota being rounded up or down; (2) \emph{Committee Monotonicity}, if $k$ increases then the quotas do not decrease; and (3) \emph{Monotonicity}, if $s_i<s_j$ and the 
quotas are perturbed such that the percentage increase of $s_i$ is more than the percentage increase of $s_j$, then $i$ should not lose a slot to $j$.
We call discrete quota allocations that satisfy the quota rule and allocate exactly $k$ slots as \emph{nice allocations}. 


\subsection{Curse of Determinism: The Need for Randomization}


We first demonstrate that randomization is a necessary feature of an apportionment mechanism in order to select exactly $k$ agents and strategyproofness. Therefore, any deterministic and strategyproof method of using fractional quotas to derive integer quotas can result in outcomes that are not of the target size (e.g.,~\cite{ALM+16a}).

%


\begin{theorem}\label{deterministicRoundImpossibility}
No Partition-based method, which assigns non-integer quotas to each cluster can select exactly $k$ agents by rounding the quotas in a deterministically strategyproof way.
\end{theorem}


We first prove the following lemma.


\begin{lemma}\label{sharesSingleValue}
In a deterministic strategyproof allocation mechanism that selects $k$ agents from $\ell$ clusters, the number of agents chosen from cluster $i$ with a share $s_{i}$ will not change, regardless of the rest of the shares.
\end{lemma}

\begin{proof}
Let $s=(s_{1},\ldots,s_{\ell})$ be the shares for each cluster, under which a mechanism allocates $y$ agents from cluster $i$. Now, let $s'=(s'_{1},\ldots,s'_{i-1},s_{i},s'_{i+1},\ldots,s'_{\ell})$ be a different share allocation. We wish to show that the number of agents selected from cluster $i$ remains $y$.

We know $\sum_{j\neq i}s_{j}=\sum_{j\neq i}s'_{j}$. If $\sum_{j\neq i}|(s_{j}-s'_{j})|\leq\frac{2}{n}$, this means a single agent can cause the change from $s$ to $s'$. As the share values do not contain data on actual agents votes, that single agent could be in cluster $i$ (since $s_{i}$ did not change at all). Thanks to strategyproofness, this means there is no change in the number of agents selected from cluster $i$---it is still $y$ agents. 

If $\sum_{j\neq i}|(s_{j}-s'_{j})|>\frac{2}{n}$, we make the move from $s$ to $s'$ using intermediary steps, $s^{0},\ldots,s^{h}$ such that $a^{j}=(s^{j}_{1},\ldots,s^{j}_{\ell})$ a share allocation where $s^{j}_{i}=s_{i}$, $s^{0}=s$, $s^{h}=s'$, and for $1\leq t\leq h$, $\sum_{j\neq i}|(s^{t}_{j}-s^{t-1}_{j})|\leq \frac{2}{n}$. Thanks to the argument in the previous paragraph, the number of agents selected from cluster $i$ stays $y$ in $s^{1}$. Now we can look at $s^{1}$ and $s^{2}$ by themselves, and due to the same argument, the number of agents from cluster $i$ needs to be the same in $s^{1}$ and $s^{2}$, hence it is still $y$ in $s^{2}$. We apply this argument again and again, until we reach the point where the number of agents selected from cluster $i$ in $s^{h}=s'$ is $y$ as well.
\end{proof}

\begin{proof}[Proof of Theorem~\ref{deterministicRoundImpossibility}]
Suppose there is a rounding of quotas that guarantees the selection of $k$ agents. Let us assume $k$ clusters and $k>3$ is odd. Using Lemma~\ref{sharesSingleValue}, we know that each cluster's slot allocation is fixed according to its share, regardless of other clusters' share. Hence, for each cluster with a share of $1.5$, it receives an allocation of either $1$ slot or $2$.

\paragraph{Case I: There are 2 clusters that are allocated $2$ slots when their share is $1.5$}

Suppose these 2 clusters have a share of $1.5$, some other cluster has share $0$, and all remaining clusters have share $1$. Hence we had $k$ shares, but received $k+1$ slot allocations.

\paragraph{Case II: There are 2 clusters that are allocated $1$ slot when their share is $1.5$}

Suppose these 2 clusters have a share of $1.5$, some other cluster has share $0$, and all remaining clusters have share $1$. Hence we had $k$ shares, but received $k-1$ slot allocations.
\end{proof}

Determinism does not just prevent strategyproof mechanisms, but also anonymous ones as shown by the following Theorem.



\begin{theorem}\label{deterministicMechImpossibility}
No Partition-based method, which assigns non-integer quotas to each cluster can select exactly $k$ agents by rounding the quotas in a deterministically anonymous way.

\end{theorem}
\begin{proof}
Let $\ell=3$, with clusters being of equal size. All agents in cluster $1$ rank agents in cluster $2$ before any in cluster $3$; agents in cluster $2$ rank those in cluster $3$ ahead of cluster $1$; and those in cluster $3$ rank agents in cluster $1$ ahead of those in $2$. So share of each cluster is equal, and for $k<\ell$ there is no deterministic anonymous way to allocate quotas.
%
%
\end{proof}

\subsection{A Novel Randomized Apportionment Rule}

We resort to randomization to achieve \emph{ex ante} fairness and the target size. Using randomization, our goal is to ensure that the target number of total agents chosen is exactly $k$ \emph{ex post}.
 Therefore we will require that the integer quota allocation returned by the lottery is a nice allocation. 
%
Randomization has been used in various settings such as voting and fair allocation of indivisible goods to achieve \emph{ex ante} as well as \emph{procedural} fairness~\citep{BCKM12a,BoMo01a,Gibb77a}. 

One possible way to achieve the appropriate randomization is to enumerate all the feasible discrete quota allocations and then solve equations to find the probability distribution over these quota allocations. If such a probability distribution exists, the method outlined involves enumerating an exponential number of such quota allocations that is computationally infeasible if the number of groups is large (for example, in U.S.~elections, $\ell$ is 50). Hence, some suggested approaches to this problem require multiple rounds of randomization and also do not enumerate the possible \emph{ex post} outcomes (there may be an exponential number of them)~\citep{RKPS06a}.
%
%
Previously, a stochastic apportionment rule was presented that achieves the quota requirements~\cite{Grim04a} by two randomizations, one of them using a stochastic continuous variable. This means that it does not involve a probability distribution over discrete nice allocations, hence it cannot be used to achieve fairness via repeated representation. Moreover, due to computers not being able to reproduce truly continuous values, this may compromise strategyproofness. 

In view of these challenges, we 
present a simple method \allocshares that achieves the target quotas, relies on a probability distribution over a \emph{linear number} of nice allocations, the probability distribution can be computed in linear time, and requires minimal randomization (only one round). Our randomized procedure can be easily de-randomized in repetitive settings. In frequently repeated allocation settings, one could use the nice allocations computed by \allocshares (at most $\ell$, compared to the potentially exponential number) in a way such that the allocations have the same frequency as the probability distribution computed by the algorithm. In this sense, our randomized apportionment routine has an advantage over other proposed methods.

Informally, our method proceeds gradually from quotas which need to be rounded up with low probability, while keeping an eye on our two main constraints: not rounding up a quota too much, on the one hand, while not being left with not enough probability for allocations with quotas that need to be rounded up with high probability.

		\begin{algorithm}[ht!]
		 \caption{\allocshares $(s_{1},\ldots,s_{\ell})$}
		 \label{algo:DP2}
		\renewcommand{\algorithmicrequire}{\wordbox[l]{\textbf{Input}:}{\textbf{Output}:}}

		 \renewcommand{\algorithmicensure}{\wordbox[l]{\textbf{Output}:}{\textbf{Output}:}}
		\algsetup{linenodelimiter=\,}
		\begin{algorithmic}
			\REQUIRE A real-value allocation $(s_{1},\ldots,s_{\ell})$ over $\ell$ objects.
			\ENSURE A discrete allocation $(t_{1},\ldots,t_{\ell})$ over $\ell$ objects.
		\end{algorithmic}
		 \begin{algorithmic}[1] 
			 \STATE Sort and renumber $(s_{1},\ldots, s_{\ell})$ according to size of $s_{i}-\floor{s_{i}}$, with $s_{1}-\floor{s_{1}}$ being minimal.
			 \STATE Let $(p_{1},\ldots,p_{\ell})\longleftarrow (0,\ldots,0)$ where $p_i$ is the probability of rounding up cluster $i$.
			 \STATE Let $\bar{p}\longleftarrow 0$, the total probability allocated so far.
			 \STATE Let $D \longleftarrow \emptyset$, where D maps: allocation $\rightarrow$ probability.
			 \STATE $\alpha\longleftarrow\sum_{i=1}^{\ell}(s_{i}-\floor{s_{i}})$
			 \STATE $low\longleftarrow 1$; $high\longleftarrow \ell$

			 \WHILE{$low\leq high$}

			 \STATE Let $allocation\longleftarrow (\floor{s_{1}},\ldots,\floor{s_{low-1}},\ceil{s_{low}},\ldots,$
			 \STATE \hskip\algorithmicindent\relax $\ceil{s_{low+\alpha-1}},\floor{s_{low+\alpha}},\ldots,\floor{s_{high}}$,$\ceil{s_{high+1}},\ldots,\ceil{s_{\ell}})$
			 \STATE \hskip\algorithmicindent\relax Where if low=1, we start with $\ceil{s_{1}}$; if $high=\ell$, we
			 \STATE \hskip\algorithmicindent\relax end with $\floor{s_{high}}$; and if $\alpha=0$, we have only $\floor{s_{low}}$.
			 \STATE $prob\longleftarrow 0$
			 \STATE $prevLow\longleftarrow low$; $prevHigh\longleftarrow high$

			 \IF{$\alpha=0$}
			 \STATE $prob\longleftarrow 1-\bar{p}; high\longleftarrow high-1$\label{alphaZero}

			 \ELSE \IF{$s_{low}-\floor{s_{low}}-p_{low}<\ceil{s_{high}}-s_{high}-\bar{p}+p_{high}$}\label{probBounds}
			 \STATE $prob\longleftarrow s_{low}-\floor{s_{low}}-p_{low}$; $low\longleftarrow low+1$\label{lowGrows}
			 \ELSE
			 \STATE $prob\longleftarrow \ceil{s_{high}}-s_{high}-\bar{p}+p_{high}$
			 \STATE $high\longleftarrow high-1$\label{highShrinks}; $\alpha\longleftarrow \alpha-1$
			 \ENDIF
			 \ENDIF
			 \FORALL{$i$ such that $prevLow\leq i<prevLow+\alpha$ or $prevHigh<i\leq \ell$}
			 \STATE $p_{i}\longleftarrow p_{i}+prob$
			 \ENDFOR
			 \STATE $\bar{p}\longleftarrow \bar{p}+prob$
			 \STATE $D \longleftarrow D \cup (allocation \rightarrow prob)$
			 \ENDWHILE
			 \STATE Select an allocation $(t_{1},\ldots,t_{\ell})$ according to $D$.
\RETURN $(t_{1},\ldots,t_{\ell})$

		 \end{algorithmic}
		\end{algorithm}

\begin{example}\label{allocationExample}
We give an example of the working of \allocshares, shown in Algorithm \ref{algo:DP2}, when a set of quotas are not integers.
Suppose we have the following Dollar shares for $\ell=5$:
$$\vec{s} = (1.1,2.1,1.3,1.7,1.8)$$
We wish to select $k=8$ agents. The $\alpha$, number of clusters that need to be rounded up, computed on line 5, is 2, and we start with $low=1;high=5$.

\noindent
We begin by considering the allocation:
$$\vec{t}_1 =(\ceil{1.1},\ceil{2.1},\floor{1.3},\floor{1.7},\floor{1.8}) = (2,3,1,1,1)$$
Since $0.1=s_{1}-\floor{s_{1}}<\ceil{s_{5}}-s_{5}=0.2$, this allocation will get a probability of $0.1$. Now $low=2$ (since cluster $1$ should not be rounded up any more), $\bar{p}=0.1$, and we ``slide'' our allocation by one to the right, and look at the next allocation:
$$\vec{t}_2 =(\floor{1.1},\ceil{2.1},\ceil{1.3},\floor{1.7},\floor{1.8}) = (1,3,2,1,1)$$
Since the second cluster has been rounded up with a probability of $0.1$ in the previous allocation, $s_{2}-\floor{s_{2}}-p(v_{2})=0$. Therefore this allocation is given probability $0$, $\bar{p}$ does not change and now $low=3$. We now move to allocation:
$$\vec{t}_3 =(\floor{1.1},\floor{2.1},\ceil{1.3},\ceil{1.7},\floor{1.8}) = (1,2,2,2,1)$$
We see $0.3=s_{3}-\floor{s_{3}}-p_{v_{3}}>\ceil{s_{5}}-s_{5}-\bar{p}+p_{v_{5}}=0.1$, so this allocation is given probability $0.1$, $\bar{p}=0.2$, and from now on cluster $5$ will always be rounded up in every allocation we consider. Hence, $\alpha$ and $high$ now change: $\alpha=1$ and $high=4$. We now turn to look at:
$$\vec{t}_4 =(\floor{1.1},\floor{2.1},\ceil{1.3},\floor{1.7},\ceil{1.8}) = (1,2,2,1,2)$$
Since $0.2=s_{3}-\floor{s_{3}}-p_{v_{3}}=\ceil{s_{4}}-s_{4}-\bar{p}+p_{v_{4}}=0.2$, we give this allocation the probability $0.2$, $\bar{p}=0.4$, and $low=4$. Finally, we look at:
$$\vec{t}_5 =(\floor{1.1},\floor{2.1},\floor{1.3},\ceil{1.7},\ceil{1.8}) = (1,2,1,2,2)$$
We give this allocation the probability $s_{4}-\floor{s_{4}}-p_{v_{4}}=0.6$.

Overall, the algorithm yields a probability distribution over $\ell$ allocation vectors.
 \begin{equation*}
 \begin{split}
 t_1 = (2,3,1,1,1)&: 0.1\\
 t_2 = (1,3,2,1,1)&: 0.1\\
 t_3 = (1,2,2,2,1)&: 0.1\\
 t_4 = (1,2,2,1,2)&: 0.2\\
 t_5 = (1,2,1,2,2)&: 0.6\\
 \end{split}
 \end{equation*}
\end{example}

This defines our probability space, and the expected number of agents selected from each cluster is exactly its Dollar share: $(1.1,2.1,1.3,1.7,1.8)$.
\begin{theorem}\label{allocationWorks}
\allocshares defines a distribution and the expected allocation of each cluster $i$ is its share $s_i$.
\end{theorem}

To prove this theorem, we first need several lemmas. Note that any cluster $i$ needs to be rounded up with probability of $s_{i}-\floor{s_{i}}$, and rounded down with probability $\ceil{s_{i}}-s_{i}$.

\begin{lemma}\label{alwaysAlpha}
Let $z=(z_{1},\ldots,z_{\ell})\in \mathbb{N}^{\ell}$ 
and let set $A=\{z'\in\mathbb{N}^{\ell}~ | ~\text{in}~ \alpha\in\mathbb{N} ~\text{coordinates $z'_{i}=z_{i}+1$. In the rest $z'_{i}=z_{i}$}\}$. For any $\hat{z}=(\hat{z}_{1},\ldots,\hat{z}_{\ell})$ that is a simplex of $A$ (i.e., $\hat{z}=\sum_{a\in A}p_{a}a$ such that $\sum_{a\in A}p_{a}=1$), $\sum_{i=1}^{\ell}(\hat{z}_{i}-z_{i})=\alpha$.
\end{lemma}

\begin{proof}

\begin{equation*}
\sum_{i=1}^{\ell}(\hat{z}_{i}-z_{i})=\sum_{i=1}^{\ell}\big[(\sum_{a\in A}(p_{a}a_{i}))-(\sum_{a\in A}p_{a})z_{i}\big]=\sum_{a\in A}p_{a}\big[\sum_{i=1}^{\ell}(a_{i}-z_{i})\big]
\end{equation*}

For any $a\in A$, from $A$'s definition we know $\sum_{i=1}^{\ell}(a_{i}-z_{i})=\alpha$. Hence:
$$
\sum_{a\in A}p_{a}\big[\sum_{i=1}^{\ell}(a_{i}-z_{i})\big]=\sum_{a\in A}p_{a}\alpha=\alpha\sum_{a\in A}p_{a}=\alpha
$$

\end{proof}

Note that Lemma~\ref{alwaysAlpha} is applicable to our algorithm, as we can consider $z=(\floor{s_{1}},\ldots,\floor{s_{\ell}})$ as a basis, and in each allocation the algorithm rounds up $\alpha$ coordinates, and this rounding up is equivalent to taking $\alpha$ coordinates, and instead of using the value from $z$ (the rounded down value, $\floor{s_{i}}$), we round up and add 1 to the coordinate.


\begin{lemma}\label{neverTooMuch}
At no point in the algorithm is a cluster rounded up or down in allocations that, together, have more probability than it should, i.e., for any cluster $i$, it is always true that $p_{i}\leq s_{i}-\floor{s_{i}}$ and $\bar{p}-p_{i}\leq \ceil{s_{i}}- s_{i}$.
Moreover, for any $i<low$, the probability that cluster $i$ is rounded up is $s_{i}-\floor{s_{i}}$. For any $i>high$, the probability that cluster $i$ is rounded down is $\ceil{s_{i}}-s_{i}$.
\end{lemma}
\begin{proof}
Recall that the probability cluster $i$ is rounded up needs to be $s_{i}-\floor{s_{i}}$. When a set of clusters is being rounded up, the probability of the allocation is bounded by $s_{low}-\floor{s_{low}}-p_{low}$ (Line~\ref{probBounds}), i.e., the probability $s_{low}$ should be rounded up which still remains to be allocated. Any other cluster $i$ rounded up in the same allocation has been rounded with $s_{low}$ in every previous allocation when it has been rounded up (since $i>low$), so $p_{i}\leq p_{low}$. Since the clusters are ordered according to $s_{i}-\floor{s_{i}}$ in line 1, we know that $s_{i}-\floor{s_{i}}>s_{low}-\floor{s_{low}}$. Hence, $s_{low}-\floor{s_{low}}-p_{low}\leq s_{i}-\floor{s_{i}}-p_{i}$, so no cluster is rounded up more than it should be.

At any point in the algorithm, if $i<low$, then there was a stage where $low=i$, and line~\ref{lowGrows} changed $low$ to $i+1$. However, at that point, cluster $i$ is rounded up exactly the additional probability it needed to be rounded up ($s_{i}-\floor{s_{i}}-p_{i}$) , and no further allocation in the algorithm will round it up.

Similarly, if $i>high$ (and $\alpha\neq 0$), there was a stage where $high=i$, and line~\ref{highShrinks} changed $high$ to $i-1$. However, at that point, cluster $i$ is rounded down exactly the additional probability it needs to be rounded down. It could not have been rounded down too much previously, as every allocation's probability is bounded so that cluster $high$ will not be rounded down more than it is supposed to ($\ceil{s_{high}}-s_{high}$). 
Once again, once an index is $high+1$, no further allocation will round it down.

For cluster $i$, $low\leq i\leq high$, it has only been rounded down when $s_{high}$ was rounded down as well (though not vice versa) and rounded up when $s_{low}$ was rounded up (again, not vice versa), and as the $\ceil{s_{i}}-s_{i}\geq\ceil{s_{high}}-s_{high}$ and $s_{i}-\floor{s_{i}}\geq s_{low}-\floor{s_{low}}$, these clusters have not been rounded up more than $s_{i}-\floor{s_{i}}$, or rounded down more than $\ceil{s_{i}}-s_{i}$. Therefore, the sum of allocation probabilities is never larger than 1 (since $s_{i}-\floor{s_{i}}+\ceil{s_{i}}-s_{i}=1$). Hence, also, for clusters $i<low$, as they have received the exact needed probability of being rounded up, and since the sum of allocations does not exceed $1$, they have not been rounded down more than needed.
\end{proof}

\begin{proof}[Proof of Theorem~\ref{allocationWorks}]
We first show all the allocations we consider only round up $\sum_{i=1}^{\ell}(s_{i}-\floor{s_{i}})$ clusters, as otherwise, allocations are not allocating exactly $k$ agents. As long as $low+\alpha\leq high$ in the algorithm this is trivially true. We now wish to show the situation $low+\alpha>high$ cannot happen (as that results in too few clusters rounded up).

If $low+\alpha>high$ then there was a stage in which $low+\alpha=high$, and then we executed line~\ref{lowGrows}. This means that cluster $high$ needs more unallocated probability to be rounded down than cluster $low$ needs unallocated probability to be rounded up. Moreover, thanks to the monotonicity of elements of the $s_{i}$ vector, we know cluster $high$ still has need for more allocations with positive probability in which it is rounded up. But this property means that if we advanced $low$ and assigned all remaining unassigned probability to the allocation rounding up clusters $low+1,\ldots,\ell$, we would be rounding them up too much, and for clusters $low+1,\ldots,high$, strictly so. But we know from Lemma~\ref{neverTooMuch} that for all $i\leq low$, we have rounded the share $s_i$ up exactly correctly, so looking at the vector $\hat{z}\in\mathbb{N}^{\ell}$ in which each coordinate is the expected allocation for that cluster, we have:
\begin{equation*}
\begin{split}
\sum_{i=1}^{\ell}(\hat{z}_{i}-\floor{s_{i}})=&\sum_{i=1}^{low}(s_{i}-\floor{s_{i}})+\sum_{i=low+1}^{\ell}\hat{z}_{i}-\floor{s_{i}}>\\
>&\sum_{i=1}^{\ell}(s_{i}-\floor{s_{i}})=\alpha
\end{split}
\end{equation*}
This contradicts Lemma~\ref{alwaysAlpha}, as all allocations had exactly $\alpha$ clusters rounded up. So it cannot be that cluster $high$ still needed more probability to be rounded down, and therefore, if $low+\alpha=high$, line~\ref{lowGrows} would not have been executed at this point.

Since $low+\alpha\leq high$ at all times, the algorithm will end when $low=high$. Hence, the previous step ended with $\alpha$ becoming $0$ in line~\ref{highShrinks}. Observe that $\alpha=0$ only in this case: otherwise, it means the sum of expected value---that is, probability to be rounded up---over all clusters is above $\alpha$: we have too many clusters that need to be rounded up. 

We now wish to prove that the last step, where the clusters $high+1,\ldots,\ell$ are rounded up, results in what we desired. Since clusters $1,\ldots,low-1$ have been allocated the right probability to be rounded up, as well as clusters $high+1,\ldots,\ell$ (Lemma~\ref{neverTooMuch}), we only need to verify this for cluster $low$. But according to Lemma~\ref{alwaysAlpha}, the probability of $low$ being rounded up is exactly $\alpha-\sum_{1\leq i\leq \ell, i\neq low}(s_{i}-\floor{s_{i}})$, which is exactly $s_{low}-\floor{s_{low}}$, which means cluster $low$ has the correct allocation.
\end{proof}

\section{Analytical Comparisons with Other Mechanisms}\label{sec:comparision}

Though \edp draws inspiration from Dividing a Dollar and Partition, there are key differences between these mechanisms and clear reasons to use \edp over other potential variants.

%
%
%
%

\subsection{Comparison with other Dollar Based Mechanisms}
Although \edp is partly based on the Dollar mechanism for dividing a bonus (division of a divisible item between agents), it is more desirable than some other natural mechanisms one can construct based on the Dollar framework. Consider the following possible adaptations of the Dollar framework and their shortcomings. 
\begin{description}[itemsep=0.1cm,leftmargin=0.25cm]
\item[Dollar Raffle:] Take the dollar mechanism (without any partitions), compute the relative fraction of the dollar each agent should receive. Use these fractions as a probability distribution over the agents and then repeatedly select an agent according to its dollar share until $k$ different agents are selected. 

\item[Dollar Partition Raffle:] Take the Dollar shares of the \emph{clusters} in Dollar Raffle and use these shares to define a probability distribution over the clusters. A cluster is drawn with respect to the cluster Dollar probabilities and the next best agent, based on reviews of agents outside the cluster, is selected, until $k$ different agents are selected. 

\item[Top Dollar:] Select the agents with maximum Dollar shares.\footnote{Vanilla is equivalent to Top Dollar when agents' valuations are normalized.}

\end{description}

Both Dollar Raffle and Dollar Partition Raffle have a non-zero probability of selecting the $k$ worst agents. While Top Dollar is not strategyproof for any $k<n$,
Dollar Raffle and Dollar Partition Raffle are strategyproof for $k=1$. None, however, are strategyproof for $n>k>1$. 

\begin{theorem}\label{raffleNoStrategyproof}
	Dollar Raffle, Dollar Partition Raffle, and Top Dollar are not strategyproof for $n>k>1$.
\end{theorem}
\begin{proof}
	


For Dollar Raffle and Dollar Partition Raffle, the proof follows a similar path: The mechanism iterates until it chooses $k$ different agents, which is equivalent to eliminating each selected agent and re-normalizing the dollar partitions, i.e., the probabilities of being selected, since once some agent is selected we ignore its repeated selection. This re-normalization prevents the mechanism from being strategyproof, as now the probabilities of others matter for each agent.

For example, an agent will prefer to contribute to a very strong agent. This strong agent, once eliminated, will make our agent's probability increase significantly. Suppose $k=2$ using Dollar Raffle (Dollar Partition Raffle), and suppose all agents (clusters) except $b_{1},b_{2},b_{3}$ allocate their points equally between those 3. $b_{1}$ divides its point equally between $b_{2}$ and $b_{3}$, as does $b_{2}$ between $b_{1}$ and $b_{3}$. Suppose $b_{3}$ believes it should also divide its point equally between $b_{1}$ and $b_{3}$. In that case, it has a probability $\frac{1}{3}$ of being selected first, and a probability of $\frac{1}{3}$ of being selected second, ultimately, $\frac{2}{3}$. But if agent (in) $b_{3}$ decides to give its point fully to (cluster) $b_{1}$, the probability of $b_{3}$ being selected first does not change. But the probability of $b_{1}$ being selected and then $b_{3}$ is $\frac{1}{3}$ (in Dollar Partition Raffle: $(\frac{1}{3}+\frac{1}{2n})\frac{\frac{1}{3}}{\frac{2}{3}-\frac{1}{2n}}$), and the probability of $b_{2}$ being selected and then $b_{1}$ is $\frac{1}{15}$ (in Dollar Partition Raffle: $(\frac{1}{3}-\frac{1}{2n})\frac{\frac{1}{3}}{\frac{2}{3}+\frac{1}{2n}}$). The sum of these is more than $\frac{1}{3}$, hence doing so would improve agent (in cluster) $b_{3}$ chances of being selected. This proof can easily be extended to any additional $k$.

For Top Dollar, agents are $a_{1},\ldots, a_{k+1}$. Agents $a_{1},\ldots a_{k-1}$ allocate each of their points by giving $\frac{1}{k}-\frac{1}{k^{2}}$ to agent $a_{k+1}$, $\frac{1}{k^{2}}$ to agent $a_{k}$ and $\frac{1}{k}$ to all other agents. Agent $a_{k}$ gives $\frac{1}{k}$ to all other agents. Agent $a_{k+1}$ would like to allocate its point to agent $a_{k}$, but that would mean it would not be selected itself. Giving its point to other agents will mean it will be, contradicting strategyproofness.
\end{proof}

Interestingly, the proof of this theorem for Dollar Raffle and Dollar Partition Raffle carries on, quite straightforwardly, to the various mechanisms presented for $k=1$ (e.g., \cite{FeKl14a}). Simply running the algorithm several times destroys its strategyproofness. This is true even for mechanisms that are strategyproof for $k=1$, as long as any agent has the power to influence the outcome, i.e., not purely random, a dictatorship, or a combination of both.

\subsection{Comparison with Partition Mechanisms}
\edp seems similar to the Partition mechanism but while Partition must preset the number of agents to be selected from each cluster, \edp \emph{relies on the peer reviews} to decide the number of agents to be selected from each cluster. This difference allows \edp to have more consistent performance, no matter the clustering. Hence, in contrast to \edp, if a particularly bad partition is chosen at random, the rigidity of Partition means that it may not choose a large proportion of the best agents even if agents have unanimous valuations.

\begin{example}
Consider the setting in which $N=\{1,\ldots, 18\}$, $k=6$, and $\ell=3$. Let the clusters be $C_1=\{1,\ldots, 6\}$, $C_2=\{7,\ldots, 12\}$, $C_3=\{13,\ldots, 18\}$. $C_1$ puts all its weight on $C_2$, equally dividing its points between $7,8,\ldots,12$, with a slight edge to $7$ and $8$, $C_2$ and $C_3$ put all the weight on $C_1$, dividing their points between $1,2,3$ and $4$. Now Partition will choose $1,2,7,8,13,14$ where everyone thinks that $1,2,3,4,7,8$ are the best. \edp will select exactly that set. Moreover, if we increase the number of clusters, the disparity between \edp and Partition only grows.
\end{example}
	
Partition, in contrast to \edp, performs poorly \emph{ex post}\footnote{For high stakes outcomes, we want a mechanism that performs well on average and rarely returns an especially bad outcome.} if the clusters are lopsided, with some cluster containing all good agents and other clusters containing low value agents. One natural fix is to deliberately choose a balanced partition where the weight of a cluster is based on the ratings of agents outside the cluster and we choose a clustering that minimizes the difference between the cluster weights.
However, for this and various notions of balanced partitions, computing the most balanced partition is NP-hard. What is even more problematic is that if we choose a balanced partition, the resulting mechanism is not strategyproof.

We point out that there are instances where Partition may perform better than \edp even if the rankings of the agents are unanimous. Consider a case where a highly preferred agent is in the same group as the lowest preferred agents, while other groups only contain medium preferred agents. In that case the weight of the cluster with the highest preferred agent might be so high that the lowest ranked agents might also be selected.
The normalization of scores entailed in \edp causes a certain loss of information and granularity compared to the other mechanisms.
However, even in the example above, \edp will ensure that when agents have highly correlated or unanimous preferences, the agent(s) that are unanimously on the top will be selected, even if some low-ranked agents are also selected.

\section{Experimental Comparison with Other Mechanisms}

Using Python and extending code from \textsc{PrefLib} \cite{MaWa13a} we have implemented the \edp, Credible Subset, Partition, Dollar Raffle, Dollar Partition Raffle, and Vanilla peer selection mechanisms. All the code developed for this project is available open-sourced in the PeerSelection repository on GitHub.\footnote{\textsc{https://github.com/nmattei/peerselection}} As in all simulations there are many parameters to consider that can drastically affect the outcome (see e.g.,~\cite{PRM13a}). By focusing on a target domain, the NSF Mechanism Design Pilot \cite{MeSa09a, NSF14a}, we can draw focused conclusions from our simulations. In 2014, this program had $n=131$ proposals, with each submitter reviewing $m=7$ other proposals, broken into $\ell=4$ clusters. The acceptance numbers are not broken out from the global $\approx$20\% acceptance rate, so we use this as the acceptance rate.


\subsection{Experimental Setup}


To create NSF-like data we generate a sparse $N = \{1, \ldots n\}$ scoring matrix (profile) using a \emph{Mallows Model} to generate the ordinal evaluation \cite{Mall57a,Mard96a}. Mallows models are parameterized by a \emph{reference order ($\sigma$)} and a \emph{dispersion parameter ($\phi$)}. The reference order $\sigma$ can be thought of as the underlying ground truth; the NSF mechanism assumes implicitly that there is some true ordering of proposals of which the reviewers provide a noisy observation. Intuitively, the dispersion parameter $\phi$ is the probability of committing ranking errors by swapping neighboring elements of $\sigma$, where $\phi=0$ means that no agent ever commits an error and $\phi=1.0$ means that orderings are generated uniformly at random \cite{LuBo11a}. Mallows models are used when each agent is assumed to have the same reference ranking subject to some independent noise from a common noise model. Each agent $i \in N$ ends up with a ranking $rank(i,j) \rightarrow \{0, \ldots, m-1\}$ over each agent $j \in N$ where $rank(i,j)=0$ means $j$ is the highest reviewed proposal by $i$. In our evaluation we assume that all agents share the same ground truth ranking, $\sigma$, and that all agents share the same noise parameter $\phi$ that we sweep across a range of values. An interesting direction for future work would be comparing the algorithms when agents may have different notions of the ground truth ordering, i.e., different values for $\sigma$, and/or have different levels of ability, i.e., different values for $\phi$ (as \cite{CKV16a,CKV15a} implicitly do).

Each agent reviews $m$ of the $n$ proposals and is also reviewed by $m$ other agents. Agents are clustered under the constraint that each agent reviews $m$ agents \emph{outside} his cluster. We call a reviewer assignment satisfying these constraints a balanced $m$-regular assignment. To maximize inter-cluster comparison, we want the $m$ reviews provided by agent $i$ to be balanced among the clusters (less $C_i$) so agent $i$ in cluster $C_i$ reviews in total $\frac{m}{\ell -1}$ agents from each other cluster. We generate this assignment randomly and as close to balanced as possible. Up to this point, our setup is similar to that of \citet{CKV16a}, used for studying grade aggregation in MOOCs.

We generate a sparse $n \times n$ score matrix by: drawing a balanced $m$-regular assignment; generating a complete ordinal ranking using a Mallows model for agent $i$; removing all candidates from $i$'s ordinal ranking not assigned to $i$; and assigning score $m-rank(i,j)$ to each agent $j$ that $i$ ranks (known as the Borda score). This process mimics the underlying assumption of the NSF mechanism in that, if a reviewer were to see all proposals, they could strictly order the complete set. The Borda score is well-motivated in this setting as it is the optimal scoring rule, i.e., returns a result closest to the ground truth ranking, for aggregation when agents submit correct orderings \cite{CKV16a}, and is the one used by the NSF in their pilot. This process leaves us with a sparse $n \times n$ score matrix which obeys an $m$-regular assignment of agents partitioned into $\ell$ clusters. 

We use two different orderings to evaluate the performance of all the mechanisms presented: the ground truth (GT) ordering and the ordering selected by vanilla (V). Since a Vanilla-like mechanism is used in many settings, it is fitting to see how well the strategyproof mechanisms approximate it (though we assume it is not manipulated by the participants despite not being strategyproof). This allows us to understand the ``price'' we are paying for strategyproofness. Formally, let $W, W'$ be the winning sets returned by two mechanisms, we measure the similarity of $W$ to $W'$ by $|W \cap W'|/k$. For $n=130$ and $s=1000$, we looked at an ``NSF Like'' space with $k \in \{15, 20, 25, 30, 35\}$, $m \in \{5, 7, 9, 11, 13, 15\}$, $\phi \in \{0.0, 0.10, 0.20, 0.35, 0.50\}$, and $\ell \in \{3, 4, 5, 6\}$.


\subsection{General Results}
It is hard to directly compare results for Credible Subset due to the high probability of returning an empty set under the given parameters. In fact, counter to intuition, Credible Subset performs worse as we increase the number of reviews because this \emph{increases} the chance of returning an empty set (see, e.g., Figure \ref{fig:results}.) This problem is not easy to overcome; removing the ability to return an empty set means Credible Subset is no longer strategyproof. When Credible Subset does return a set, it performs very well, on par with Vanilla. However, the minimum (0 in some cases), average, and standard deviation for Credible Subset are all unacceptable for practical implementation.

\begin{figure}[!htbp]
\centering
\makebox[\textwidth][c]{\includegraphics[width=.65\paperwidth]{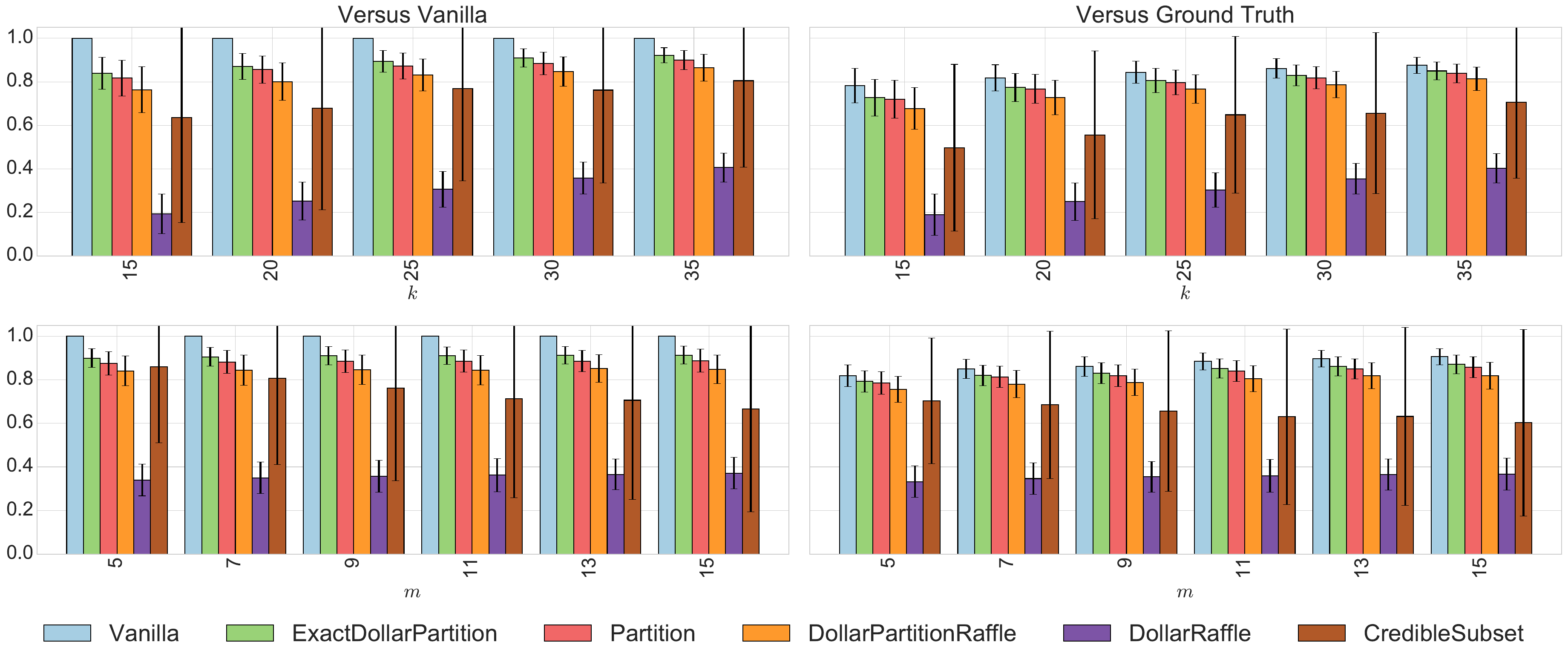}}
\caption{Performance of the six mechanisms surveyed in this paper as we vary settings to $k$ (top row) and $m$ (bottom row) as measured in selected percentage of the Vanilla set (V, left) and ground truth ordering (GT, right). Note that the colors of the bars are in the same order as the table key. For both figures we have set $n=130$, $\ell=4$, and $\phi=0.5$; for the top row, we set $m=9$ as we varied $k$ and for the bottom row we set $k=30$ as we vary $m$. Across the tested range \edp outperforms all other mechanisms when measured against the V or GT ordering with a lower standard deviation. As we increase the value of $m$ or $k$, \edp improves its performance at a faster rate and more consistently than any other mechanism.
}
\label{fig:results}
\end{figure}

Throughout the testing \edp strictly outperforms, for all parameter settings, the other Dollar-based mechanisms described in Section \ref{sec:comparision}. Hence, we conclude that the extra steps developed for \edp are necessary and result in a dramatically increased performance. In this discussion, we will focus on comparing the performance of \edp and Partition only, as these two mechanisms are the top two performing mechanisms that are also strategyproof. Broadly, we find that \edp outperforms Partition. On average, \edp selects between 0.5\% and 5\% more top agents, measured against V or GT. It does this with between 3\% and 25\% lower standard deviation, and selects as many or, in the most extreme cases, up to five more top performing agents. This improvement in the worst case means that \edp is an improvement of up to 25\% when the clustering of the agents is lopsided. Hence, \edp is the best in our experiments: it selects more top agents, more often, with better worst case performance and lower variance than any other strategyproof mechanism.

\subsection{Varying Noise($\phi$) and Clusters ($\ell$)}
Varying the noise parameter $\phi$ has little to no effect on the approximation to V for all mechanisms. This is not surprising as all mechanisms receive the same (noisy) information. When approximating GT, the value of $\phi$ has a negligible effect on the performance of the mechanisms unless $\phi \geq 0.95$; for the remainder, of the discussion we fix $\phi = 0.5$. Varying the setting of $\ell$ we see that all mechanisms perform best with respect to GT when we set $\ell = 5$ and with respect to V when $\ell = 3$; with a decrease in performance as we continue to increase $\ell$. No matter the setting, increasing the number of clusters hurt the performance of Vanilla and \edp the least, i.e., their performance decreases less quickly than the other mechanisms. For the remainder we set $\ell = 4$ as was done for the NSF pilot.



\subsection{Varying the Number of Selections ($k$) and Reviews ($m$)}

Figure~\ref{fig:results} captures our metrics as we vary the number of selections $k$, in the top row, and the number of reviews per item $m$, in the bottom row. Varying the setting to $k$ we observe fairly consistent performance by the mechanisms with \edp maintaining a 1.5\% to 3\% advantage. The biggest percentage-wise advantages are found when $k=15$, where \edp selects up to two more top agents according to V, in the worst case, resulting in a $\approx 25\%$ improvement. In the worst case, up to two more top agents according to GT are selected by \edp than Partition. Because both mechanisms perform worse (in absolute terms) than they do as measured by V, this translates to a 10--20\% increase in performance for \edp when $k$ is small and a 5--10\% increase when $k$ is large. Measured against both V and GT we can draw the general conclusion that, as we increase $k$, \edp increases its advantage over Partition.

For the most NSF-like setting where we have $n=130, k=30, l=4, \phi=0.5$, we sweep $m \in \{5,7,9,11,13,15\}$, depicted as the bottom row of Figure \ref{fig:results}. Looking closely at the numbers as measured against V we see that \edp performs 2.6 to 3.0\% better on average, i.e., one better agent. \edp does this with a nearly 20\% smaller standard deviation, always selecting at least one more and up to three more top agents (15\%) in the worst case. This pattern is similar across settings to the other parameters as vary $m$; \edp performs consistently better as we increase the number of reviews. Compared to the performance of Vanilla, \edp selects about one more non-top agent on average, up to two more non-top agents in the worst case ($\approx 7\%$). Hence, the loss in performance we see for moving to a strategyproof mechanism is similar to the loss in performance we see when moving to Partition from \edp.

\section{Conclusion}

The problem we have considered here is one that has many common applications: from NSF funding allocations and conference paper selection to voting for a committee in an organization's board and decision making within groups. All these problems are, fundamentally, a set of peers selecting the ``best'' subset of themselves according to their own quality criteria.
We detail a new strategyproof mechanism, which incorporates ideas from the Partition mechanism \cite{AFPT11a} and literature on dividing a continuous resource \cite{CMT08a,TiPl08a}, combined with a new allocation mechanism, which addresses a long-standing problem of turning a fraction allocation into an integer one, while adding some desirable properties over existing solutions. Moreover, we are able to show, via a set of simulations, that our proposed mechanism performs better than other existing mechanisms.

The next stage in this line of research, we believe, will not have to do with finding additional strategyproof mechanisms, but rather with finding ways to eliminate problematic agents' preferences. This might be achieved either by relaxing the notion of strategyproofness in return for a degree of agent incentive, or by identifying ``problematic'' agents or very good ones, whose opinions may be weighed differently.

\section*{Acknowledgments}
Authors wish to thank Allan Borodin, Markus Brill, Manuel Cebrian, Serge Gaspers, Ian Kash, Julian Mestre, and Herv{\'{e}} Moulin for useful comments. 
Data61/CSIRO (formerly known as NICTA) is funded by the Australian Government through the Department of Communications and the Australian Research Council through the ICT Centre of Excellence Program. This research has also been partly funded by Microsoft Research through its PhD Scholarship Program, Israel Science Foundation grant grants \#1227/12 and \#1340/18, and NSERC grant 482671. This work has also been partly supported by COST Action IC1205 on Computational Social Choice. Haris Aziz was supported by a Julius Career Award and a UNSW Scientia Fellowship.

%






\end{document}